\documentclass[11pt,openany]{amsart}
\def\isdraft{0}
\usepackage{mathtools}
\usepackage{amsmath,amssymb,amsfonts,dsfont}
\usepackage{prettyref} 
\usepackage[utf8]{inputenc}
\usepackage{enumitem}
\usepackage[hyphens]{url}
\usepackage{tikz-cd}
\usepackage{tikz}
\usetikzlibrary{positioning}
\usepackage{bussproofs}
\usepackage[mathscr]{euscript}
\usepackage{hyperref}
\usepackage{amsthm}
\usepackage{theapa}
\usepackage{a4wide}
\usepackage[color=black,textcolor=white,disable]{todonotes}
\if\isdraft 1
\usepackage[color=black,textcolor=white]{todonotes}
\fi
\usepackage{stackengine}
\usepackage{stmaryrd}

\newtheorem{theorem}{Theorem}

\newtheorem{proposition}[theorem]{Proposition}
\newtheorem{lemma}[theorem]{Lemma}

\theoremstyle{definition} 

\newtheorem{example}[theorem]{Example}

\newtheorem{problem}[theorem]{Problem}

\newrefformat{cha}{Chapter \ref{#1}}
\newrefformat{sec}{Section \ref{#1}}
\newrefformat{tab}{Table \ref{#1}}
\newrefformat{fig}{Figure \ref{#1}}
\newrefformat{equ}{(\ref{#1})}
\newrefformat{app}{Appendix \ref{#1}}
\newrefformat{thm}{Theorem \ref{#1}}
\newrefformat{cor}{Corollary \ref{#1}}
\newrefformat{prop}{Proposition \ref{#1}}
\newrefformat{lem}{Lemma \ref{#1}}
\newrefformat{fact}{Fact \ref{#1}}
\newrefformat{obs}{Observation \ref{#1}}
\newrefformat{note}{Note \ref{#1}}
\newrefformat{idea}{Idea \ref{#1}}
\newrefformat{trivia}{Trivia \ref{#1}}
\newrefformat{def}{Definition \ref{#1}}
\newrefformat{not}{Notation \ref{#1}}
\newrefformat{con}{Convention \ref{#1}}
\newrefformat{rem}{Remark \ref{#1}}
\newrefformat{exa}{Example \ref{#1}}
\newrefformat{problem}{Problem \ref{#1}}
\newrefformat{claim}{Claim \ref{#1}}
\newrefformat{conjecture}{Conjecture \ref{#1}}
\newrefformat{exe}{Exercise \ref{#1}}
\newrefformat{alg}{Algorithm \ref{#1}}
\newrefformat{err}{Error \ref{#1}}
\newrefformat{que}{Question \ref{#1}}
\newrefformat{ite}{Item \ref{#1}}
\newrefformat{Q}{Q. \ref{#1}}
\newrefformat{warning}{Warning \ref{#1}}
\newrefformat{pseudocode}{Pseudocode \ref{#1}}

\title{Logic-based similarity}
\author{
	Christian Anti\'c\\
}

\begin{document}
\begin{abstract} 
	This paper develops a {\em qualitative} and logic-based notion of similarity from the ground up using only elementary concepts of first-order logic centered around the fundamental model-theoretic notion of type.
\end{abstract}

\maketitle

\section{Introduction}

Detecting and exploiting similarities between seemingly distant objects is at the core of analogical reasoning which itself is at the core of artificial intelligence with applications to such diverse tasks as proving mathematical theorems and building mathematical theories, commonsense reasoning, learning, language acquisition, and story telling \cite<e.g.>{Boden98,Gust08,Hofstadter01,Hofstadter13,Krieger03,Polya54,Winston80,Wos93}.

This paper develops a {\em qualitative} and logic-based notion of similarity {\em from the ground up} using only elementary tools of first-order logic centered around the fundamental model-theoretic notion of type. The idea is to say that two elements, possibly from different domains, are similar iff they share a maximal set of abstract properties. 

We show that the so-obtained notion of similarity has appealing mathematical properties (\prettyref{sec:Properties}). Specifically, we show in our First and Second Isomorphism Theorems \ref{thm:FIT}, \ref{thm:SIT} that similarity is compatible with isomorphisms, which is in line with \citeS{Gentner83} famous structure-mapping theory of analogical reasoning.

In a broader sense, this paper is a further step towards a mathematical theory of analogical reasoning \cite<cf.>{Antic22,Antic22-3,Antic23-4,Antic22-4,Antic_i7,Antic22-1}.

\section{Similarity}

This is the main section of the paper. Here we shall introduce from first principles a notion of similarity based on the well-known notion of model-theoretic type. We expect the reader to be fluent in first-order logic as it is presented for example in \citeA{Hinman05}. 

In the remainder of the paper, we fix a first-order language $L$, and two $L$-structures $\mathbb A$ and $\mathbb B$.

Recall that the type of an element collects all the properties of that element in a structure. Comparing two elements from different structures can thus be achieved by comparing their respective types. Formally, the {\em type} of $a$ in $\mathbb A$ is usually defined by
\begin{align*} 
	Type_\mathbb A(a):=\{\varphi(y)\mid\mathbb A\models\varphi(a)\}.
\end{align*} Intuitively, the type of $a$ contains all abstract properties of $a$ in $\mathbb A$---in a sense, it is the theory of $\mathbb A$ from $a$'s perspective. Thus we can interpret the set $T(a,b):=Type_\mathbb A(a)\cap Type_\mathbb A(b)$ as the set of shared abstract properties of $a$ and $b$. We thus expect the similarity of $a$ and $b$ to increase with the number of formulas in $T(a,b)$ and it appears reasonable to say that $a$ and $b$ are similar iff $T(a,b)$ is maximal. However, there are three problems. 

First, if we require $T(a,b)$ to be maximal with respect to $a$ and $b$ {\em at the same time}, we clearly have $T(a,b)\subseteq T(a,a)$ and $T(a,b)\subseteq T(b,b)$ and in most cases this inclusion is proper which means that $a$ and $b$ would be similar iff $a=b$ which is too strong to be useful---we therefore separate the maximality conditions and say that $T(a,b)$ is maximal among the $T(a,b')$, $b'\in\mathbb A$, {\em and} maximal among the $T(a',b)$, $a'\in\mathbb A$, respectively. 

Second, since we always have $T(a,b)\subseteq T(a,a)$, we run into the same problems as before---we therefore require $T(a,b)$ to be maximal with respect to $b$ {\em except for $a$}, and maximally with respect to $a$ {\em except for $b$}.

Third, if we allow arbitrary formulas to appear in the types of $a$ and $b$, then we have $(y=a)\lor (y=b)\in T(a,b')$ iff $b'=b$ or $b'=a$, and we have $(y=a)\lor (y=b)\in T(a',b)$ iff $a'=a$ or $a'=b$, which means that with the previously adapted maximality condition, again $a$ and $b$ are similar iff $a=b$, which is undesirable. We therefore restrict types to contain only formulas without disjunctions (which forces us to restrict negation as well).

This motivates the following definitions. We call an $L$-formula {\em conjunctive} iff it contains no disjunctions and negation occurs only in front of atoms to form literals. We denote the set of all conjunctive $L$-formulas by $c\text-Fm_L$. 

The following definition is an adaptation of the above model-theoretic notion of type to conjunctive formulas. Define the {\em conjunctive $L$-type} (or {\em $c$-type}) of an element $a\in\mathbb A$ by
\begin{align*} 
	c\text-Type_\mathbb A(a):=\{\varphi(y)\in c\text-Fm_L \mid \mathbb A\models\varphi(a)\}.
\end{align*}

We are now ready to formally introduce the main notion of the paper. Given $a\in\mathbb A$ and $b\in\mathbb B$, we define the {\em similarity} relation by\footnote{In the rest of the paper, we will abbreviate ``is $\subseteq$-maximal with respect to $b$ except for $a$ in case $a\in\mathbb B$'' by ``is $b$-maximal''.}
\begin{align*} 
	(\mathbb{A,B})\models a\lesssim b \quad\text{iff}\quad \text{$c\text-Type_\mathbb A(a)\cap c\text-Type_\mathbb B(b)$ is {\em $\subseteq$-maximal} with respect to $b$}\\ 
		\text{except for $a$ in case $a\in\mathbb B$}.
\end{align*} We define
\begin{align*} 
	(\mathbb{A,B})\models a \approx b \quad\text{iff}\quad (\mathbb{A,B})\models a\lesssim b \quad\text{and}\quad (\mathbb{B,A})\models b\lesssim a,
\end{align*} in which case we say that $a$ and $b$ are {\em similar} in $(\mathbb{A,B})$. 

We extend similarity from elements to algebras by
\begin{align*} 
	\mathbb A\lesssim\mathbb B \quad\text{iff}\quad \text{for each $a\in\mathbb A$ there exists some $b\in\mathbb B$ such that $(\mathbb{A,B})\models a\approx b$,}
\end{align*} and
\begin{align*} 
	\mathbb A\approx\mathbb B \quad\text{iff}\quad \mathbb A\lesssim\mathbb B \quad\text{and}\quad \mathbb B\lesssim\mathbb A.
\end{align*} 

For convenience, we define
\begin{align*} 
	c\text-Type_{(\mathbb{A,B})}(a\lesssim b):=c\text-Type_\mathbb A(a)\cap c\text-Type_\mathbb B(b),
\end{align*} and we call the formulas in that set {\em justifications} of $a\lesssim b$ in $(\mathbb{A,B})$. Notice the symmetry
\begin{align*} 
	c\text-Type_{(\mathbb{A,B})}(a\lesssim b)=c\text-Type_{(\mathbb{B,A})}(b\lesssim a).
\end{align*} 

Of course,
\begin{align}\label{equ:subseteq} 
	c\text-Type_\mathbb A(a)\subseteq c\text-Type_\mathbb B(b) \quad\text{implies}\quad (\mathbb{A,B})\models a\lesssim b.
\end{align} and 
\begin{align*} 
	c\text-Type_\mathbb A(a)=c\text-Type_\mathbb B(b) \quad\text{implies}\quad (\mathbb{A,B})\models a\approx b.
\end{align*}

\subsection*{Characteristic justifications}

Computing all justifications of a similarity is difficult in general which fortunately can be omitted in many cases.

We call a set $J$ of justifications a {\em characteristic set of justifications} of $a \lesssim b$ in $(\mathbb{A,B})$ iff $J$ is a sufficient set of justifications, that is, iff
\begin{enumerate}
\item $J\subseteq c\text-Type_{(\mathbb{A,B})}(a\lesssim b)$, and
\item $J\subseteq c\text-Type_{(\mathbb{A,B})}(a\lesssim b')$ implies $b'=b$, for each $b'\neq a\in\mathbb B$.
\end{enumerate} In case $J=\{\varphi\}$ is a singleton set satisfying both conditions, we call $\varphi$ a {\em characteristic justification} of $a \lesssim b$ in $(\mathbb{A,B})$.

\begin{example} Take for example the natural numbers $\mathbb N$ starting at $0$ and the set $2^B$ of all subsets of some set $B$ with the usual operations. Then it is reasonable to expect $0$ to be similar to the empty set $\emptyset$. Given the joint language $L$ containing a function symbol $\ast$ standing for addition and union, respectively, we see that the (conjunctive) formula
\begin{align*} 
	\varphi(y):\equiv (\forall z)(z\ast y=z)
\end{align*} is a characteristic justification of the similarity
\begin{align}\label{equ:0_approx_emptyset} 
	(\mathbb N,2^B)\models 0 \approx \emptyset
\end{align} since
\begin{align*} 
	\varphi\in c\text-Type_{(\mathbb N,2^B)}(0,B') \quad\text{iff}\quad B'=\emptyset,
\end{align*} and
\begin{align*} 
	\varphi\in c\text-Type_{(\mathbb N,2^B)}(a,\emptyset) \quad\text{iff}\quad a'=0.
\end{align*}

Similarly, the (conjunctive) formula (here $\oast$ stands for multiplication and intersection)
\begin{align*} 
	\psi(y):\equiv (\forall z)(z\oast y=y)
\end{align*} is another characteristic justification of the similarity in \prettyref{equ:0_approx_emptyset}.
\end{example}

\section{Properties}\label{sec:Properties}

In this section, we prove some elementary properties of similarity.

\begin{theorem}\label{thm:properties} In a single structure, the similarity relation is reflexive, symmetric, and in general not transitive. In a pair of structures, reflexivity may fail as well.
\end{theorem}
\begin{proof} Symmetry holds trivially. 

In a single structure, reflexivity holds trivially as well. To disprove reflexivity in a pair of structures $(\mathbb{A,B})$, it suffices to find elements $a\in\mathbb{A\cap B}$ and $b\in\mathbb B$ so that
\begin{align*} 
	c\text-Type_{(\mathbb{A,B})}(a\lesssim a)\subsetneq c\text-Type_{(\mathbb{A,B})}(a\lesssim b),
\end{align*} and it is not hard to construct such counterexamples. For example, in the two structures $\mathbb A$ and $\mathbb B$ given, respectively, by
\begin{center}
\begin{tikzpicture} 
	\node (0) {$0$};
	\node (1) [right=of 0] {$1$};
	\node (2) [right=of 1] {$2$};
	\node (3) [right=of 2] {$\ldots$};
	\draw[->] (0) to [edge label={$f$}] (1);
	\draw[->] (1) to [edge label={$f$}] (2);
	\draw[->] (2) to [edge label={$f$}] (3);
\end{tikzpicture}
\end{center} and
\begin{center}
\begin{tikzpicture} 
	\node (0) {$1$};
	\node (1) [right=of 0] {$0$};
	\node (2) [right=of 1] {$2$};
	\node (3) [right=of 2] {$\ldots$};
	\draw[->] (0) to [edge label={$f$}] (1);
	\draw[->] (1) to [edge label={$f$}] (2);
	\draw[->] (2) to [edge label={$f$}] (3);
\end{tikzpicture}
\end{center} we have
\begin{align*} 
	c\text-Type_{(\mathbb{A,B})}(1\lesssim 1)\subsetneq c\text-Type_{(\mathbb{A,B})}(1\lesssim 0)
\end{align*} which shows
\begin{align*} 
	(\mathbb{A,B})\models 1\not\lesssim 1.
\end{align*}

To disprove transitivity, consider the structures $\mathbb{A,B,C}$ given respectively by
\begin{center}
\begin{tikzpicture}[node distance=1cm and 2cm] 
	\node (a) {$a$};
	\node (a') [above=of a] {$a'$};
	\node (b) [right=of a] {$b$};
	\node (b') [above=of b] {$b'$};
	\node (c) [right=of b] {$c$};
	\node (c') [above=of c] {$c'$};
	\draw[->] (a) to [edge label={$R$}] (a');
	\draw[<->] (b) to [edge label={$R$}] (b');
	\draw[->] (c') to [edge label={$R$}] (c);
\end{tikzpicture}
\end{center} It is easy to verify
\begin{align*} 
	(\mathbb{A,B})\models a\lesssim b \quad\text{and}\quad (\mathbb{B,C})\models b\lesssim c \quad\text{whereas}\quad (\mathbb{A,C})\models a\not\lesssim c
\end{align*} since
\begin{align*} 
	c\text-Type_{(\mathbb{A,C})}(a\lesssim c)\subsetneq c\text-Type_{(\mathbb{A,C})}(a\lesssim c').
\end{align*}
\end{proof}

\begin{problem}\label{problem:transitivity} Characterize those structures in which the similarity relation {\em is} reflexive and transitive.
\end{problem}

\section{Isomorphism Theorems}

It is reasonable to expect isomorphisms---which are {\em bijective} structure-preserving mappings between structures---to be compatible with similarity. Consider the following simple example. Let $\Sigma:=\{a\}$ be the alphabet consisting of the single letter $a$, and let $\Sigma^\ast$ denote the set of all words over $\Sigma$ including the empty word $\varepsilon$. We can identify every sequence $a^n=a\ldots a$ ($n$ consecutive $a$'s) with the non-negative integer $n$, for every $n\geq 0$. Therefore, define the isomorphism $F:(\mathbb N,+)\to (\Sigma^\ast,\cdot)$ via
\begin{align*} 
    F(0):=\varepsilon \quad\text{and}\quad F(n):=a^n,\quad n\geq 1.
\end{align*} We expect the following similarities to hold:
\begin{align*} 
    ((\mathbb N,+),(\Sigma^\ast,\cdot))\models n \approx F(n),\quad\text{for all }n\geq 0.
\end{align*} That this is indeed the case is the content of the First Isomorphism \prettyref{thm:FIT} below.

\begin{lemma}[Isomorphism Lemma]\label{lem:IL} For any isomorphism $F:\mathbb{A\to B}$,
\begin{align*} 
	c\text-Type_\mathbb A(a)=c\text-Type_\mathbb B(F(a)),\quad\text{for all $a\in\mathbb A$}.
\end{align*}
\end{lemma}
\begin{proof} A straightforward structural induction proof \cite<cf.>[Lemma 2.3.6]{Hinman05}.
\end{proof}

\begin{theorem}[First Isomorphism Theorem]\label{thm:FIT} For any isomorphism $F:\mathbb{A\to B}$,
\begin{align*} 
	(\mathbb{A,B})\models a\approx F(a),\quad\text{for all $a\in\mathbb A$.}
\end{align*}
\end{theorem}
\begin{proof} A direct consequence of the Isomorphism \prettyref{lem:IL}.
\end{proof}

The following counterexample shows that similarity is, in general, {\em not} compatible with homomorphisms.

\begin{example}\label{exa:H} Let $\mathbb A:=(\{a,b\},f)$ and $\mathbb B:=(\{c\},f)$ be given by
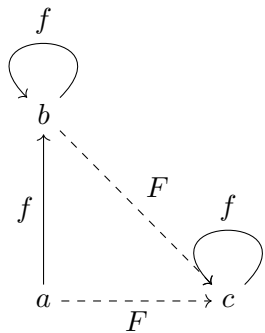
\begin{center}
\begin{tikzpicture}[node distance=2cm and 2cm]
\node (b)  {$b$};
\node (a) [below=of b] {$a$};
\node (c) [right=of a] {$c$};

\draw[->] (a) to [edge label={$f$}] (b); 
\draw[->] (b) to [edge label'={$f$}] [loop] (b);
\draw[->] (c) to [edge label'={$f$}] [loop] (c);
\draw[dashed,->] (a) to [edge label'={$F$}] (c);
\draw[dashed,->] (b) to [edge label={$F$}] (c);
\end{tikzpicture}
\end{center} The mapping $F:\mathbb{A\to B}$ is obviously a homomorphism. However, we clearly have
\begin{align*} 
	c\text-Type_\mathbb A(a)\cap c\text-Type_\mathbb B(c)\subsetneq c\text-Type_\mathbb A(b)\cap c\text-Type_\mathbb B(c)=c\text-Type_\mathbb B(c),
\end{align*} which shows
\begin{align*} 
	(\mathbb{B,A})\models F(a)\not\lesssim a
\end{align*} and thus
\begin{align*} 
	(\mathbb{A,B})\models a\not\approx F(a).
\end{align*}
\end{example}



\begin{theorem}[Second Isomorphism Theorem]\label{thm:SIT} For any isomorphims $F:\mathbb{A\to C}$ and $G:\mathbb{B\to D}$,
\begin{align*} 
	(\mathbb{A,B})\models a\approx b \quad\Leftrightarrow\quad (\mathbb{C,D})\models F(a)\approx G(b).
\end{align*}
\end{theorem}
\begin{proof} A direct consequence of
\begin{align*} 
	c\text-Type_\mathbb A(a)=c\text-Type_\mathbb C(F(a)) \quad\text{and}\quad c\text-Type_\mathbb B(b)=c\text-Type_\mathbb D(F(b)).
\end{align*}
\end{proof}

\if\isdraft 1
\section{Relational structures}

For each $n\geq 1$, define
\begin{align*} 
	\chi_n:\equiv (\exists z_1\ldots\exists z_n)\bigwedge_{1\leq i\neq j\leq n}(z_i\neq z_j).
\end{align*} The formula $\chi_n$ says that there are at least $n$ distinct elements.

\begin{proposition} For any $a,b\in\mathbb N$,
\begin{align*} 
	(\mathbb N,\leq)\models a\lesssim b \quad\text{iff}\quad a\leq b
\end{align*} which implies
\begin{align*} 
	(\mathbb N,\leq)\models a\approx b \quad\text{iff}\quad a=b.
\end{align*}
\end{proposition}
\begin{proof} 
\todo[inline]{}
\end{proof}
\fi

\section{Future work}

The major line of future research is to extend the concepts and results of this paper from first-order to second-order and higher-order logic \cite<cf.>{Leivant94}.

Another major line of future work is to study logic-based similarity in concrete structures relevant in practice as for example in the setting of words as studied in computational linguistics and natural language processing. Since computing the set of all (characteristic) justifications and conjunctive types is difficult in general, it is reasonable to study fragments of our framework first by syntactically restricting the form of justifications (cf. \prettyref{sec:Generalization-basedimilarity}).

\section{Related work}

There is a vast literature on similarity, most of which deals with {\em quantitative} approaches where similarity is numerically measured in one way or another; some works dealing with the similarity of words are \citeA{Hirschberg75} and \citeA{Egho15}.\todo{refs} A notable exception is \citeA{Yao00} where qualitative judgements of similarity are interpreted through statements of the form ``$a$ is more similar to $b$ than $c$ to $d$.'' \citeA{Badra18} analyze the role of qualitative similarity in analogical transfer.

\subsection{Generalization-based similarity}\label{sec:Generalization-basedimilarity}

A notable fragment of the above notion of similarity is due to \citeA{Antic23-2}. For this, we associate with every $L$-term $s(\mathbf z)$ a {\em $g$-formula} of the form
\begin{align*} 
	\varphi_{s(\mathbf z)}(y):\equiv(\exists\mathbf z)(y=s(\mathbf z)).
\end{align*} We denote the set of all $g$-formulas over $L$ by $g\text-Fm_L$. Notice that every $g$-formula is a conjunctive formula which means that we obtain a notion of {\em $g$-similarity}\footnote{The name is derived from the observation that we can interpret every term $s(\mathbf z)$ satisfying $a=s(\mathbf o)$, for some $\mathbf o\in\mathbb A^{|\mathbf z|}$, as a {\em generalization} of $a$ in $\mathbb A$.}
\begin{align*} 
	(\mathbb{A,B})\models a\approx_g b
\end{align*} by replacing $c\text-Type$ by $g\text-Type$, which is based on $g$-formulas, in the definition of similarity. This coincides with the definition of generalization-based similarity in \citeA{Antic23-2}.

\section{Conclusion}

This paper developed a notion of qualitative logic-based similarity from the ground up using only elementary notions of first-order logic centered around the fundamental notion of type. In a broader sense, this paper is a further step towards a mathematical theory of analogical reasoning.

\bibliographystyle{theapa}
\bibliography{/Users/christianantic/Bibdesk/Bibliography,/Users/christianantic/Bibdesk/Preprints,/Users/christianantic/Bibdesk/Publications}
\if\isdraft 1

\section{Unary algebras}

\fi
\end{document}